%% file: main.tex
\documentclass[sigconf]{acmart}
\input{preamble}

\usepackage[ruled]{algorithm2e}
\SetKwProg{For}{for}{}{}
\SetKwProg{If}{if}{}{}
\usepackage[capitalize]{cleveref}
\title{Positivity Proofs for Linear Recurrences with Several Dominant Eigenvalues}
\author{Alaa Ibrahim}
\affiliation{\institution{INRIA, ENS de Lyon, LIP}
\city{Lyon}
\country{France}}

\begin{document}

\copyrightyear{2025} \acmYear{2025}
\acmConference[ISSAC
'25]{International Symposium on Symbolic and Algebraic
Computation}{2025}{Guanajuato, Mexico}
\acmBooktitle{International Symposium on Symbolic and Algebraic
Computation (ISSAC '25), 2025, Guanajuato, Mexico}
\begin{abstract}
Deciding the positivity of a sequence defined by a linear recurrence and initial conditions is, in general, a hard problem. 

When the coefficients of the recurrences are constants, decidability has only been proven up to order 5. The difficulty arises when the characteristic polynomial of the recurrence has several roots of maximal modulus, called dominant roots of the recurrence.

We study the positivity problem for recurrences with polynomial coefficients, focusing on sequences of Poincaré type, which are perturbations of constant-coefficient recurrences. The dominant eigenvalues of a recurrence in this class are the dominant roots of the associated constant-coefficient recurrence.

Previously, we have proved the decidability of positivity for recurrences having a unique, simple, dominant eigenvalue, under a genericity assumption. The associated algorithm proves positivity by constructing a positive cone contracted by the recurrence operator.

We extend this cone-based approach to a larger class of recurrences, where a contracted cone may no longer exist. The main idea is to construct a sequence of cones. Each cone in this sequence is mapped by the recurrence operator to the next. This construction can be applied to prove positivity by induction.

For recurrences with several simple dominant eigenvalues, we provide a condition that ensures that these successive inclusions hold. Additionally, we demonstrate the applicability of our method through examples, including recurrences with a double dominant eigenvalue.




\end{abstract}
\maketitle{}

\section{Introduction}
A sequence $(u_n)_{n\in \mathbb{N}}$ of real numbers is called \emph{P-finite} of \emph{order} $d$ if it satisfies a linear recurrence of the form
 \begin{equation}\label{rec}
    p_d(n)u_{n+d}=p_{d-1}(n) u_{n+d-1}+\dots+p_0(n)u_n,\qquad n\in \mathbb{N},
 \end{equation}
with coefficients $p_i\in\mathbb R[n]$. The sequence $(u_n)_n$ is said to be positive if each term of the sequence is non-negative, i.e., $u_n\geq 0$ for all $~n\in\mathbb{N}$.

When the coefficients $p_i$ are constants in $\mathbb{R}$, the sequence is called C-finite. 
In this work, we take the assumption $0\not\in p_d(\mathbb{N})$. Then, the sequence is completely defined by the recurrence relation and the first $d$ terms of the sequence, $(u_0,u_1,\dots,u_{d-1})$. 

The \emph{Positivity Problem} consists in deciding whether a sequence $(u_n)_n$ is positive, given the polynomials $(p_i)_{i=1}^d$ and initial conditions. In this work, we consider the positivity problem over rationals: $p_i\in\mathbb{Q}[n]$ and $(u_0,\dots,u_{d-1})\in\mathbb{Q}^d$.

For instance, the following inequality on hypergeometric functions \cite{f2a34dbf-37e5-3f3a-8559-d3d2b37f29f8} $$f(x)={}_2F_1(\frac{1}{2}, \frac{1}{2}; 1; x^2)- (1+x) ~{}_2F_1(\frac{1}{2}, \frac{1}{2}; 1; x)\geq 0,~x\in(0,1),$$
follows from the positivity of the coefficients $(a_n)_n$ of the Taylor expansion of $f(x)=\sum_{n\geq 0} a_nx^n$. Thus, proving the inequality reduces to deciding the positivity of a P-finite sequence whose order can be further reduced to 2, as shown in \cref{sec:Examples}.


This reduction is partially due to closure properties of P-finite sequences under addition, product and Cauchy product.  Furthermore, for any \(\ell \in \mathbb{N}_{>0}\) and \(q \in \{0, \dots, \ell-1\}\), the subsequence \((u_{\ell n+q})_{n\in\mathbb{N}}\) satisfies a linear recurrence of order at most \(d\). All these operations are effective, meaning that recurrences for the resulting sequences can be computed given recurrences for the input~\cite{Stanley1999}.

Thus, various problems can be reduced to positivity. For example, proving monotonicity ($u_{n+1}\ge u_n$), convexity ($u_{n+1}+u_{n-1}\ge 2u_n$), log-convexity ($u_{n+1}u_{n-1}\ge u_n^2$) or more generally an inequality with another P-finite sequence ($u_n\ge v_n$), can be interpreted as instances of the positivity problem. Moreover, the reduction of the famous Skolem problem to Positivity \cite{ibrahim2024positivityproofslinearrecurrences} highlights the hardness of the positivity problem.

The Positivity Problem arises in multiple domains. In mathematics, it appears in numerous inequalities within combinatorics \cite{mitrinovic1964elementary, ScottSokal2014}, number theory \cite{StraubZudilin2015} and the theory of special functions  ~\cite{Pillwein2008}.  In computer science, positivity occurs in several areas such as the study of numerical stability in computations involving sums of power series \cite{SerraArzelierJoldesLasserreRondepierreSalvy2016} and the verification of loop termination \cite{HumenbergerJaroschekKovacs2017,HumenbergerJaroschekKovacs2018}. Other applications can be found in floating-point error analysis \cite{BoldoClementFilliatreMayeroMelquiondWeis2014} and in biology \cite{MelczerMezzarobba2022,YuChen2022,BostanYurkevich2022a}.\\

\textbf{Decidability:} The \emph{dominant eigenvalues} of the recurrence play a fundamental role in deciding positivity. For C-finite sequences, the starting point for proving positivity is the closed-form expression of the sequence. Given the characteristic polynomial of the recurrence  
\[
\chi(x) = x^d - \frac{p_{d-1}}{p_d}x^{d-1} - \dots - \frac{p_0}{p_d},
\]  
with distinct roots \(\lambda_1, \dots, \lambda_k\), it is well known that the sequence takes the form  
\[
u_n = C_1(n) \lambda_1^n + \dots + C_k(n) \lambda_k^n,
\]  
where each \(C_i(n)\) is a polynomial determined by the initial conditions \(u_0, \dots, u_{d-1}\).  

In this representation, the roots \(\lambda_i\) with the largest modulus govern the asymptotic behavior of the sequence; they are also called the \emph{dominant eigenvalues}. In all previous works, deciding positivity for C-finite sequences has been addressed through discussions on the dominant eigenvalues. 
First, Ouaknine and Worrell proved the decidability for recurrences of order up to~5. This limitation is related to the number of dominant eigenvalues. Furthermore, for $d=6$ positivity is reduced to open problems in Diophantine approximation~\cite{OuaknineWorrell2014a}. For C-finite recurrences with one dominant eigenvalue, decidability is established for arbitrary order. Moreover, when the characteristic polynomial of the sequence does not have multiple roots, decidability extends to order up to~9~\cite{OuaknineWorrell2014b}. For reversible recurrences of integers (reversible means that unrolling the recurrence backwards produces only integers for negative indices), decidability of positivity is known for order up to~11 and this goes up to~17 if the recurrence is both reversible and with square-free characteristic polynomial~\cite{KenisonNieuwveldOuaknineWorrell2023}.

In summary, even in the case of C-finite sequences, the positivity problem is not fully understood. Moreover, for P-finite sequences, the situation is further complicated since their is no “simple” basis of solutions, and it is difficult to relate the asymptotic behavior of the sequence with its initial conditions even for small orders \cite{MR4309740}.\\

\textbf{Previous works:} In a very general setting, including P-finite sequences, the first algorithmic work is due to Gerhold and Kauers \cite{GerholdKauers2005}. Their method consists in looking iteratively for the first integer $m$ such that the following implication holds$$u_n\ge0\wedge u_{n+1}\ge0\wedge\dots\wedge u_{n+m}\ge0\Rightarrow u_{n+m+1}=\sum_{i=0}^m q_i(n)u_{n+i} \geq0.$$This implication is handled as a decision problem in the existential theory of the reals, and it is tested using cylindrical algebraic decomposition~\cite{collins1975quantifier}.
This method leads to automatic proofs for many important inequalities~\cite{GerholdKauers2006}, though the termination of this procedure remains unclear.

Later, Kauers and Pillwein have adapted that  method to control its termination for P-finite sequences. Positivity is concluded by proving $u_{n+1}>\beta u_n$ for some positive $\beta$ using the same approach \cite{Kauers2007a,KauersPillwein2010a}. 
This yields the decidability for recurrences of order $d=2$,  with a \emph{unique simple dominant eigenvalue}, under a \emph{genericity} assumption on the initial conditions. Under the same constraints, decidability also extends to a subclass of recurrences of order $d= 3$~\cite{Pillwein2013}. 

More recently, this decidability result was extended to arbitrary order using a cone-based approach. For recurrences having a unique simple dominant eigenvalue, positivity is proven, \emph{generically}, by constructing a positive cone contracted by the recurrence operator ~\cite{MR4687164}. Later, the construction of the cone was refined using the theory of Perron-Frobenius for cones \cite{ibrahim2024positivityproofslinearrecurrences}. The aim of the present work is to apply the same geometric approach to prove positivity for sequences in a larger class. This case was partially studied for second-order recurrences \cite{KenisonNieuwveldOuaknineWorrell2023}.

\textbf{Contribution:} We propose a procedure in \cref{algorithm:General} to prove positivity for a subclass of P-finite sequences. For recurrences of arbitrary order with multiple dominant eigenvalues, all simple, we provide in \cref{cor:decidability} asymptotic conditions on the eigenvalues of the recurrence that ensure the decidability of the positivity problem, under certain assumptions on the asymptotic behavior of the sequence. Additionally, using this procedure, we prove positivity for sequences beyond the scope of our earlier works.\\

This work is structured as follows: First, the background on dominant eigenvalues and contracted cones is reviewed in \cref{sec:Cones}. The general algorithm is presented in \cref{sec:Example}. In \cref{sec:Separation}, we show a method to verify the conditions required for the input of the given algorithm. The construction of the cones is provided in \cref{sec:construction}. The main result \cref{theo:cond} is stated in \cref{sec:theo}. Examples are given in \cref{sec:Examples}, followed by experimental improvements discussed in \cref{sec:practice}.
 \section{Contracted Cones associated to recurrences}\label{sec:Cones}
In this article, positivity is studied through a geometric approach. The recurrence relation is converted to a recurrence on vectors, and positivity is established by showing that these vectors remain in positive cones.
\subsection{Conversion to Vector Form}
Let $U_n$ be the vector defined by
\[
U_n = \trsp{(u_n, u_{n+1}, \dots, u_{n+d-1})}.
\]  
The recurrence relation in \cref{rec} is converted into the vector form
$$U_{n+1} = A(n) U_n,$$ where \( A(n) \in \mathbb{Q}(n)^{d \times d} \) has the shape of a \emph{companion matrix}.
The positivity of the sequence \( (u_n)_{n \in \mathbb{N}} \) is equivalent to ensuring that 
$$
U_n \in \mathbb{R}_{\geq 0}^d, \quad \forall n \in \mathbb{N}.$$


\subsubsection{Poincaré type} In this work, we consider P-finite sequences of Poincaré type.  They can be seen as perturbation of C-finite sequences.
\begin{definition}[Poincaré type]
The recurrence relation \cref{rec} is said of Poincaré type if the polynomial $p_d$ has maximal degree, i.e.$\lim\limits_{n\to\infty} A(n)=A\in\mathbb{Q}^{d\times d}$.
\end{definition}
\begin{remark}
The condition of being a recurrence of Poincaré type is not really a restriction, as it is always possible to reduce the positivity problem of a P-finite sequence to that of the solution of a linear recurrence of Poincaré type by re-scaling~\cite[\S2]{MezzarobbaSalvy2010}. 
\end{remark}
\begin{definition}[Dominant eigenvalues]\label{def:dominant-eigenvalues}
Let
$\lambda_1,\dots,\lambda_k$ be the distinct complex eigenvalues of the matrix $A$, numbered by decreasing modulus so that 
\[|\lambda_1|=|\lambda_2|=\dots=|\lambda_v|>|\lambda_{v+1}|\geq |\lambda_{v+2}|\dots \geq |\lambda_k|.\]
Then $\lambda_1,\dots,\lambda_v$ are called the \emph{dominant eigenvalues of $A$} (or equivalently \emph{dominant roots} of its characteristic polynomial). An eigenvalue is called \emph{simple} when it is a simple root of the characteristic polynomial.
\end{definition}
In this case, the \emph{eigenvalues of the recurrence} refer to the eigenvalues of the matrix$~A$.

\subsection{Cones}
A subset \( K \) of \( \mathbb{R}^d \) is called a \emph{cone} if it is closed under addition and multiplication by positive scalars. 
\subsubsection*{Proper cone} 
A cone \( K \) in \( \mathbb{R}^d \) is called \emph{proper} if it satisfies the following properties:
\begin{itemize}
    \item[--] \( K \) is \emph{pointed}, i.e., \( K \cap (-K) = \{0\} \).
    \item[--] \( K \) is \emph{solid}, i.e., its interior \( K^{\circ} \) is non-empty.
    \item[--] \( K \) is \emph{closed} in \( \mathbb{R}^d \).
\end{itemize}
\subsubsection*{Positive cone}
A cone \( K \) is called \emph{positive} if  \( K\subset\mathbb{R}_{\ge0}^d \).
\subsubsection*{Contracted cone}
A cone $K$ is \emph{contracted} by a real matrix $A$ if $A(K\setminus \{0\})
\subset K^{\circ}$, where $K^\circ$ denotes the interior of $K$.
\subsection{Contracted Cones}
 For recurrences having a unique simple dominant eigenvalue, positivity is \emph{generically} decidable \cite{ibrahim2024positivityproofslinearrecurrences}. This follows from Vandergraft's theorem stated below.
\begin{theorem}
[Vandergraft] \cite[Thm. 4.4]{Vandergraft1968}
\label{thm:Vandergraft}
Let  $A$  be a real matrix. There exists a proper cone $K$ contracted by  $A$  if and only if $A$ has a unique dominant eigenvalue $\lambda_{1}>0$, and $\lambda_1$ is simple.
\end{theorem}
Under these conditions, $A$ is said to satisfy the \emph{contraction condition}.

\paragraph{Application to Positivity} Assume that the limit matrix \( A \) contracts a proper cone \( K \). Then, there exists an index \( N \) such that \( A(n) K \subset K \) for all \( n \geq N \). Thus, once \( U_n \) enters \( K \), all subsequent vectors \( U_n\) remain in \( K \). This forms the basis of the \emph{Positivity Proof} algorithm in \cite{ibrahim2024positivityproofslinearrecurrences} and leads to the following result:

\begin{theorem}\cite[Thm. 1]{ibrahim2024positivityproofslinearrecurrences}\label{thm:decidability}
For all linear recurrences of the form given in \cref{rec}, of order $d$ and of Poincar\'e
type, having a unique simple dominant eigenvalue and such that
$0\not\in p_0p_d(\mathbb N)$, the positivity of the solution
$(u_n)_n$ is decidable for any $U_0=(u_0,u_1,\dots,u_{d-1})$ outside
a hyperplane in $\mathbb{R}^d$.
\end{theorem}

\section{Sequences of contracted cones}\label{sec:Example}

The principle of our approach can be seen on an example where the contraction condition is not satisfied by the limit matrix $A$.

Consider the sequence \((u_n)_n\) solution of 
\[
4(n + 3)(n + 4)^2 u_{n+2} = (n + 3)(8n^2 + 48n + 73) u_{n+1} - (n + 2)(2n + 5)^2 u_n,
\]
with initial conditions \(u_0 = \frac{1}{64}\), \(u_1 = \frac{11}{768}\).
The characteristic polynomial of the recurrence is   
\[
\chi_u(X) = (X - 1)^2.
\]  
The dominant eigenvalue is not simple, thus no proper cone is contracted by the limit matrix \( A \), preventing the application of the algorithm \emph{Positivity Proof} in \cite{ibrahim2024positivityproofslinearrecurrences}.

However, for a fixed \( n \), the matrix \( A(n) \) has a unique simple dominant eigenvalue, which implies the existence of a cone \( K_n \) that is contracted by \( A(n) \). We can take \( K_n \) to be the following proper positive cone:
\[
K_n = \left\{ a \begin{pmatrix} 3 \\ 3 + \frac{\sqrt{2} - 6}{n} \end{pmatrix} + b \begin{pmatrix} 1 \\ 1 + \frac{-2 + 3\sqrt{2}}{n} \end{pmatrix} \mid a \geq 0, b \geq 0 \right\}.
\]
\begin{figure}
    \centering
    \includegraphics[width=0.7\linewidth]{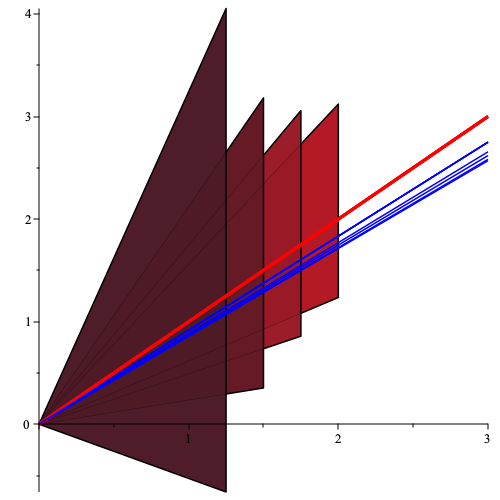}
    \caption{Cones \(K_n\) for \(n = 1\) to \(4\), where the color of \(K_n\) becomes lighter for increasing \(n\). The red vector is the limit of the cones, and the blue vectors represent \(U_0, \dots, U_4\).}
    \label{fig:ConesSequence}
\end{figure}

 It turns out that the sequence of contracted cones $(K_n)_n$ also satisfies
\[
A(n)K_n \subset K_{n+1}, \quad \text{for all } n \geq N = 3.
\]
In other words, for \(n \geq N\)
$$U_n \in K_n\implies U_{n+1} = A(n)U_n \in K_{n+1}\in\mathbb{R}^2_{>0}.$$
As shown in \cref{fig:ConesSequence}, for \(n =3\), \(U_3 \in K_3 \). By induction, $U_n\in K_n\subset\mathbb{R}^2_{>0}$ for all $n\geq N=3$. Additionally, $U_0,U_1,U_2$ are in $\mathbb{R}^2_{>0}$, thus the sequence \((u_n)_n\) is positive.

For this sequence, positivity is proved by constructing a sequence of contracted cones and induction. This approach extends to recurrences of Poincaré type where the recurrence operator $A(n)$ satisfies the contraction condition. More precisely, for $\lambda_{i,n}$ denoting the distinct eigenvalues of $A(n)$, we have
$$\lambda_{1,n}>|\lambda_{2,n}|\geq \dots \geq |\lambda_{k,n}|,\quad n\in\mathbb{N}$$
and $\lambda_{1,n}$ is simple. This approach is outlined in \cref{algorithm:General}.
\begin{algorithm}[h!]
\caption{Main Steps of the Positivity Algorithm}
\label{algorithm:General}
\KwIn{ $A(n)$ companion matrix $\in\mathbb Q(n)^{d\times d}$ such that 
\( A = \lim\limits_{n \to \infty} A(n)\in\mathbb{Q}^{d\times d} \);
$A(n)$ satisfies the contraction condition (for $n$ large enough); 
 $U_0\in\mathbb Q^d$ }
\KwOut{Positivity of $(U_n)_n$}
\begin{enumerate}
  \item[1.] \textbf{Cone Construction:} Compute
  \( K_n \subset \mathbb{R}_{\geq 0}^d \) s.t.
  \( A(n)(K_n \setminus \{0\}) \subset K_n^{\circ} \).
  
  \item[2.]\textbf{Inclusion Index:} Find \( N \in \mathbb{N} \) s.t. for all \( n \geq N \), \( A(n)K_n \subset K_{n+1} \). \\
  \item[3.] \textbf{Initial Conditions Check:} 
  Compute \( n_0 \geq N \) s.t. \( U_{n_0} \in K_{n_0} \), 
  and check if \( U_0, U_1, \dots, U_{n_0 - 1} \) are positive.
\end{enumerate}
\end{algorithm}

The termination of this algorithm is studied by observing the asymptotic behavior of the eigenvalues $\lambda_{i,n}$ as $n\to\infty$. 
\section{Eigenvalues ordering}\label{sec:Separation}
For a rational square matrix, the eigenvalues can be efficiently ordered by modulus using absolute separation bounds \cite[Theorem 1]{BugeaudDujellaFangPejkovicSalvy2022}. Let \( A(n) \) be a matrix in \( \mathbb{Q}(n)^{d\times d} \) with characteristic polynomial \( P(n,X) \), and assume that  
\[
\lim_{n\to \infty} A(n) = A \in \mathbb{Q}^{d \times d}.
\]
In the following, we show how the eigenvalues of $A(n)$, i.e. the roots of $P(n,X)$, can be ordered  by modulus for sufficiently large $n$.

First, compute the polynomial $Q(Y,X)\in\mathbb{Q}[Y,X]$ the numerator of $P(1/Y,X)$. By construction, $Q(0,X)$ is the characteristic polynomial of \( A \), up to a scalar.

By Puiseux theorem, the roots $X_j(Y)$ of \( Q(Y,X) \)  take the form  
\[
X_j(Y) = X_j + \sum_{k=1}^{\infty} c_k Y^{\frac{k}{m}}, 
\]
where \( X_j\in\mathbb{C} \) is a root of \( P(0,X) \), $m$ is a positive integer and $c_k$ is in $\mathbb{C}$. 

For sufficiently small \( Y \), the number of roots \( X_j(Y) \) with the same modulus remains constant. This allows us to define the sequence \( (\rho_i(Y))_{i=1}^{\ell} \) representing the distinct modulus of the roots \( X_j(Y) \). As $Y$ decreases to 0, we assume that $\rho_i(Y)$ have the following order
\[
\rho_1(Y) > \rho_2(Y) > \dots>\rho_{\ell }(Y).
\]
 We denote by \( q_i \) the number of roots with modulus \( \rho_i(Y)\). 
 \begin{remark}
If the $(q_i)_{i=1}^\ell$ are known, the roots $X_j(Y)$ can be grouped by modulus using Puiseux expansions up to sufficient order.
\end{remark}
\begin{proposition}
For a square free bivariate polynomial $Q(Y,X)\in\mathbb{Q}[Y,X]$, \cref{algorithm:roots} computes $(q_i)_{i=1}^{\ell}$.
\end{proposition}
\begin{proof}
\cref{algorithm:roots} is an  extension of algorithm 5 in\cite{GourdonSalvy1996}. For a square-free univariate polynomial, this algorithm computes iteratively the number of roots for each 
modulus, starting from the largest modulus to the smallest.
The proof follows the same approach as in \cite{GourdonSalvy1996}. We sketch it for the computation of $q_1$.

For $Y$ decreasing to 0, the polynomial \( Q_Y \otimes Q_Y\in \mathbb{Q}[Y,X] \) has \( \rho_1(Y)^2 \) as its unique dominant root. On the other hand, the polynomials \( P_i \) are coprime. Thus, there exists a unique index \( i_0 \) such that \( (\rho_1(Y))^2\) is a root of \( P_{i_0}\in(\mathbb{Q}[Y])[X] \) and $q_1=i_0$.

To determine \( i_0 \), \cref{algorithm:General} computes Puiseux expansions at 0 for the roots $\beta_{i,j}(Y)$, $i=1\dots \ell$. The uniqueness of $\rho_1(Y)^2$ as dominant root ensures that $\beta_{i_0,j_0}(Y)$ is found using Puiseux expansions up to a finite order. This completes the proof.
The computations of \( Q_Y \otimes Q_Y \) and \( \mathcal{P}_{m} \) can be done efficiently \cite{GourdonSalvy1996}.  
\end{proof}

\begin{algorithm}
\caption{Number of Roots for each Modulus}
\label{algorithm:roots}
\KwIn{$Q_Y(X) = \prod_{i=1}^d (X - X_i(Y))$ square free polynomial $\in\mathbb{Q}[Y][X]$ where $X_i=X_i(0)\in\mathbb{C}$.}
\KwOut{$q_i$: the number of roots of modulus \( \rho_i(Y)\), as $Y$ decreases to 0.}
\begin{enumerate}
    \item Compute \( Q_Y \otimes Q_Y = \prod_{i,j} (X- X_i(Y) X_j(Y)) \).
    \item Compute the square-free decomposition of \( Q_Y \otimes Q_Y \) in the variable \( X \):
    \[
    Q_Y \otimes Q_Y = P_0(Y) P_1(X,Y) P_2(X,Y)^2 \dotsm P_{k}(X,Y)^k.
    \]
    \item For each \( i = 1, \dots, k \), compute Puiseux expansions at $Y=0$ for the roots \( \beta_{i,j}(Y) \) of the polynomial \( P_i(X,Y) \),
    \item Iteratively, increase the order of these expansions until finding $(i_0,j_0)$ such that $\beta_{i_0,j_0}(Y)$ is a positive real root having the greatest modulus as $Y$ decreases to 0;
    \item  Set $q_1 = i_0 $.
    \item Knowing \( q_i \) for \( i = 1, \dots, j \), let $m=q_1 + q_2 + \dots + q_j + 1$. Find \( q_{j+1} \) by repeating the above steps for the polynomial
    \[
    \mathcal{P}_{m}(Y,X) = \prod_{i_1 < i_2 < \dots < i_{m}} \left( X - X_{i_1}(Y) X_{i_2}(Y) \dots X_{i_{m}}(Y) \right).
    \]
\end{enumerate}
\end{algorithm}

\section{Sequence of cones}\label{sec:construction}
For a matrix  $A$  satisfying the contraction condition, Vandergraft’s construction of the cone is based on the eigenstructure of $A$. We first recall this construction and then show how to adapt it to define a sequence of cones  $(K_n)_n$ , each contracted by  $A(n)$ .
\subsection{Vandergraft's Construction}\label{subsec:vandergraft}
We specialize Vandergraft's construction to the case where \( A \) is the companion matrix of the polynomial  
\[
P(X) = (X-\lambda_1)(X-\lambda_2)\cdots (X-\lambda_{\mu})(X-\lambda_{\mu+1})^{m_{\mu+1}}\cdots (X-\lambda_k)^{m_k}.
\]  
 Here, \( \lambda_i \) denotes an eigenvalue of \( A \), and \( m_i \) is its multiplicity.  The eigenvalues are assumed to satisfy 
\[
\lambda_1 > |\lambda_2| \geq |\lambda_3| \geq \cdots \geq |\lambda_\mu| > |\lambda_{\mu+1}| \geq \cdots \geq |\lambda_k|,
\]  
and the first $\mu\geq 2$ of them are assumed to be simple.
\subsubsection{Basis Construction}  

Fix \( \varepsilon > 0 \) such that \( |\lambda_2| - |\lambda_{\mu+1}| > \varepsilon \). If $\mu=\deg(P)$, take \( \varepsilon = 0 \).  

Construct a basis \( (V_{i,j}) \) of \( \mathbb{R}^d \) satisfying the following properties:
\begin{align}\label{eq:basis}
     A V_{i,1} &= \lambda_i V_{i,1},&\quad& 1 \leq i \leq k , \\     A V_{i,j} &= \lambda_i V_{i,j} + \varepsilon V_{i,j-1}, &\quad& \mu+1 \leq i \leq k,\ 2 \leq j \leq m_i \nonumber.
\end{align}

In other words, if \( T \) is the matrix whose columns are the vectors\\ \( (V_{i,j})_{j=1,\dots, m_i,\ i=1,\dots, k} \), then
\begin{equation}\label{eq:Jordan}
    T^{-1}AT=
\begin{bmatrix}
    \lambda_1  & \cdots & 0 & 0 & \cdots & 0 \\
    \vdots  & \ddots & \vdots & \vdots & \ddots & \vdots \\
    0  & \cdots & \lambda_\mu & 0 & \cdots & 0 \\
    0  & \cdots & 0 & J_{\varepsilon, \mu+1} & \cdots & 0 \\
    \vdots  & \ddots & \vdots & \vdots & \ddots & 0 \\
    0  & \cdots & 0 & 0 & \cdots & J_{\varepsilon, k}
\end{bmatrix}
\end{equation}
with $J_{i,\varepsilon}$ the Jordan block
\[
\qquad
J_{i,\varepsilon} =
\begin{bmatrix}
    \lambda_i & \varepsilon & 0 & \cdots & 0 \\
    0 & \lambda_i & \varepsilon & \cdots & 0 \\
    0 & 0 & \lambda_i & \ddots & \vdots \\
    \vdots & \vdots & \ddots & \ddots & \varepsilon \\
    0 & 0 & \cdots & 0 & \lambda_i
\end{bmatrix}  \in \mathbb{C}^{m_i \times m_i}.
\]
Let $C$ be the set defined by \begin{align*}
C=\Bigl\{ \alpha=(\alpha_{1,1},\dots,\alpha_{j,m_k})\in \mathbb{C}^{d\times d} \Big|&~ \alpha_{1,1} \geq |\alpha_{i,j}|,\\ \alpha_{i,j} = \overline{\alpha_{p,q}}
\text{ if } V_{i,j} = \overline{V_{p,q}},
&~\alpha_{i,j} \in \mathbb{R} \text{ otherwise} \Bigr\}.
\end{align*}
Then, 
$$\label{eq:Cone}
    K=K(V_{i,j}) = \{ T\alpha | \alpha \in C\},
$$
 is a proper real cone contracted by $A$. This is proved along  with \cref{pro:distance} below. 
 
For a vector $\beta=(\beta_1,\dots,\beta_d)$ in $\mathbb{C}^{d\times d}$, we use the norm
$$\| \beta\| =\max_{i=1}^d |\beta_i|.$$ $\Delta$ denotes the metric induced by the norm $\|\cdot\|$.
\begin{proposition}\label{pro:distance}
Let \( S_1 \subset K \) be the compact set defined by
\[
S_1 := \{T\alpha| \alpha\in C \text{ and } \alpha_{1,1} = 1\}.
\]
Then, the distance between \( T^{-1}AS_1 \) and \( T^{-1}\partial K \), where \( \partial K \) denotes the boundary of \( K \), satisfies:
\[
\Delta(T^{-1}AS_1, T^{-1}\partial K) \geq \frac{\lambda_1 - |\lambda_2|}{2} > 0.
\]

\end{proposition}
\begin{proof}\label{proof:2}
1. Action of \( A \):  
   Let \( W = T\beta \) be a vector in \( S_1 \). By definition, \( \beta_{1,1} = 1 \) and \( |\beta_{i,j}| \leq 1 \) for all \( i, j \). In view of \cref{eq:Jordan}, \( AW = T(T^{-1}AT)\beta=T\alpha \), where:
   \[
   \alpha_{1,1} = \lambda_1,
   \]
   and for \( i \neq 1 \):
   \[
   \alpha_{i,j} =
   \begin{cases}
   \beta_{i,j} \lambda_i, & j = m_i, \\
   \beta_{i,j} \lambda_i + \varepsilon \beta_{i,j+1}, & i\geq \mu+1,~j < m_i.
   \end{cases}
   \]
 2. Bounding \( |\alpha_{i,j}| \):
   Since \( |\beta_{i,j}| \leq 1 \), for \( i \neq 1 \) and $j=m_i$,
   \[
   |\alpha_{i,j}|\leq  |\lambda_i|\leq |\lambda_2|\]
For $i\geq \mu+1$ and $j<m_i$,
$$
     |\beta_{i,j}| \leq  |\lambda_i| + |\varepsilon \beta_{i,j+1}|\leq |\lambda_{\mu+1}|+\varepsilon \leq |\lambda_2|.
$$
Thus,
   \[
   (AS_1) \subset S_2 := \{T\alpha \mid \alpha_{1,1} = \lambda_1, |\alpha_{i,j}| \leq |\lambda_2| \}.
   \]
  By definition, $S_2$ is in the interior of $K$ and thus $K$ is contracted by$~A$.\\
3. Compactness of \( AS_1 \):  
   \( AS_1 \) is  a compact subset of $K^{\circ}$  because it is the image of a compact set \( S_1 \) under \( A \). \\
4. Characterization of \( \partial K \): 
   A vector \( T\alpha \) lies on the boundary of \( K \) if and only if \( \alpha_{1,1} = \displaystyle \max_{(i,j) \neq (1,1)} |\alpha_{i,j}| \), i.e.,
   \[
   \partial K =\Bigl\{ T\alpha \Big|~ \alpha \in C, \alpha_{1,1} = \max_{(i,j)\neq (1,1)} |\alpha_{i,j}|\Bigr\}.
   \]
5.\label{sec:5} Distance Calculation:
The inclusion in $S_2$ implies 
$$\Delta(T^{-1}AS_1, T^{-1}\partial K) \geq \Delta(T^{-1}S_2,T^{-1}\partial{K})=\Delta(C_1,C_2)$$ where \begin{align*}
    C_1&=\{\alpha \in C|~ \alpha_{1,1}=\lambda_1,|\alpha_{i,j}|\leq\lambda_2\},\\
 C_2&=\{\beta\in C|~ \max_{i\neq 1} |\beta_{i,j}|=\beta_{1,1}\}.
\end{align*}
By definition,
$$\Delta(C_1,C_2)=\min_{\alpha\in C_1,~\beta\in C_2}\Delta(\alpha,\beta).$$
Let $(i_0,j_0)$ be the index such that $|\beta_{i_0,j_0}|=|\beta_{1,1}|$; we have
$$\Delta(\alpha,\beta)\geq \max\Bigl\{|\lambda_1-\beta_{1,1}|,|\alpha_{i_0,j_0}-\beta_{i_0,j_0}|\Bigr\}.$$ 
For $|\beta_{1,1}|\leq \lambda_1$ or $|\beta_{1,1}|\leq|\lambda_2|$, by the inequality above we get 
$$\Delta(\alpha,\beta)\geq \lambda_1-|\lambda_2|.$$
Now, for $|\lambda_2|\leq|\beta_{1,1}|\leq \lambda_1$,
\begin{align*}
    |\lambda_1-\beta_{1,1}|&=\lambda_1-|\beta_{1,1}|\\
    |\alpha_{i_0,j_0}-\beta_{i_0,j_0}|&\geq |\beta_{i_0,j_0}|-|\alpha_{i_0,j_0}|\geq |\beta_{1,1}|-|\lambda_2|.
\end{align*}
Combining all these bounds, we get
$$\Delta(\alpha,\beta)\geq \max_{|\lambda_2|\leq |\beta_{1,1}|\leq \lambda_{1,1}}\Bigl\{\lambda_1-|\beta_{1,1}|,|\beta_{1,1}|-|\lambda_2|\Bigr\}\geq \frac{\lambda_1-|\lambda_2|}{2}.$$
Then, $\Delta(C_1,C_2)\geq \frac{\lambda_1-|\lambda_2|}{2}$. This completes the proof.
\end{proof}

\subsection{Construction of a sequence of contracted cones}
Let \( A(n) \in \mathbb{Q}(n)^{d \times d} \) be the companion matrix of the polynomial  
\[
P(n, X) = (X - \lambda_{1,n}) \cdots (X - \lambda_{v,n})(X - \lambda_{v+1,n})^{m_{v+1}} \cdots (X - \lambda_{k,n})^{m_k},
\]
where \( \lambda_{i,n} \in \mathbb{C} \) are the eigenvalues of \( A_n \), with multiplicities \( m_i \) (we can assume that \( m_i \) is constant for sufficiently large \( n \)).\\
$A(n)$ is assumed to satisfy the contraction condition, i.e. ,
\[
\lambda_{1,n} > |\lambda_{2,n}|\geq \dots \geq |\lambda_{v,n}| >|\lambda_{v+1,n}|\geq|\lambda_{v+2,n}| \dots \geq |\lambda_{k,n}|
.\]
The ordering is done as shown in \cref{sec:Separation}
\subsubsection*{Basis vectors}
First, we can fix $\varepsilon>0$, independent of $n$, such that $0<\varepsilon<|\lambda_{2,n}|-|\lambda_{v+1,n}|$ for $n$ large enough. For this choice of $\varepsilon$, the vectors $\{V_{i,j}(n)\}$ below are solutions of \cref{eq:basis}:
\begin{equation}\label{eq:solutions}\begin{aligned}
    V_{1,1}(n)&=d(1,\lambda_{1,n},\dots,\lambda_{1,n}^{d-1})\\
    V_{i,1}(n) &= (1, \lambda_{i,n}, \dots, \lambda_{i,n}^{d-1}),\quad 1\le i\le k .\\
\end{aligned}\end{equation}
For $\mu+1\le i\le k,~
    2\le j\le m_i$,

$$ V_{i,j}(n)= \varepsilon^{j-1} \left( \binom{\ell}{j-1} \lambda_{i,n}^
    {\ell - j + 1},  \ell = 0, \dots, d-1 \right).
$$
If $\lambda_{i,n}=0$, then take
\[V_{i,j}(n)=\varepsilon^{j-1}(0,\dots,0,1,0,\dots,0),\]
with~1 in the $j$th position.
\begin{proposition}\label{pro:positivity}
For $n$ large enough, the cone $K_n = K(V_{i,j}(n))$ is positive.
\end{proposition}
\begin{proof}
For simplicity, we denote $V_{i,j} = V_{i,j}(n)$ in this proof. Let $W$ be a vector in $K(V_{i,j})$. Up to scaling, we can assume that
$$ W = V_{1,1} + \sum_{i \neq 1} \alpha_{i,j} V_{i,j}, \quad |\alpha_{i,j}| \leq 1. $$

For a vector \( x \in \mathbb{C}^d \), we denote by \( (x)_l \) the \((l+1)\)-th coordinate of the vector. For \( \ell = 0, \dots, d-1 \), we have
\begin{align*}
    (W)_\ell &= (V_{1,1})_\ell + \sum_{i \neq 1} \sum_{j=1}^{m_i} \alpha_{i,j} (V_{i,j})_\ell \\
    &\geq d\lambda_1^\ell - \sum_{i \neq 1} \sum_{j=1}^{m_i} |\alpha_{i,j} (V_{i,j})_\ell| \\
    &\geq d\lambda_1^\ell - \sum_{i \neq 1} \sum_{j=1}^{m_i} |(V_{i,j})_\ell|.
\end{align*}
Now, for any fixed \( i \neq 1 \), we prove that
$$ \sum_{j=1}^{m_i} |(V_{i,j})_\ell| \leq \lambda_1^\ell.$$
By definition of the vectors \( V_{i,j} \), we have
\begin{align*}
    \sum_{j=1}^{m_i} |(V_{i,j})_\ell| &= \sum_{j=1}^{m_i} \varepsilon^{j-1} \binom{\ell}{j-1} |\lambda_i|^{\ell-j+1} \\
    &\leq (|\lambda_i| + \varepsilon)^\ell < |\lambda_1|^\ell.
\end{align*}

Thus, we have shown that \( (W)_\ell >0 \) for all \( \ell = 0, \dots, d-1 \), which implies that \( W \in\mathbb{R}^d_{>0} \). Therefore, the cone \( K_n \) is positive.
\end{proof}
\section{Several Simple Dominant Eigenvalues}\label{sec:theo}
For recurrences having several dominant eigenvalues, all simple, we now give  asymptotic condition on the eigenvalues of \( A(n) \), ensuring the existence of a sequence of contracted cones $(K_n)_n$ such that $A(n)K_n\subset K_{n+1}$ for $n$ large enough.
\begin{theorem}\label{theo:cond}
Let \( A(n) \in \mathbb{Q}(n)^{d \times d} \) be the companion matrix of the polynomial  
\[
P(n, X) = (X - \lambda_{1,n}) \cdots (X - \lambda_{v,n})(X - \lambda_{v+1,n})^{m_{v+1}} \cdots (X - \lambda_{k,n})^{m_k},
\]
where \( \lambda_{i,n} \in \mathbb{C} \) are the eigenvalues of \(A_n\), with multiplicities \( m_i \). 

Assume that the following conditions hold:
\begin{enumerate}
    \item \label{cond:1} Contraction condition: \( \lambda_{1,n} \geq |\lambda_{i,n}| \) for all \( i \neq 1 \) and $n$ sufficiently large.
    \item \label{cond:2} The eigenvalues \( \lambda_{i,n} \) converge as \( n \to \infty \) to distinct limits \( \lambda_i \in \mathbb{C} \), satisfying:
    \[
    \lambda_1 = |\lambda_2| = \cdots = |\lambda_v| > |\lambda_j| \quad \text{for } j > v.
    \]
    \item \label{cond:3} \(  \underset{i=1,\dots,k}{\max}   |\lambda_{i,n} - \lambda_{i,n+1}| = o\left(\lambda_{1,n} - \underset{j=2,\dots,v}{\max}   |\lambda_{j,n}|\right),~ n \to \infty.
    \)
\end{enumerate}
Then, there exists a sequence of proper positive cones \( K_n \subset \mathbb{R}^d \) such that 
\( A_n K_n \subseteq K_{n+1} \), for sufficiently large \( n \).
\end{theorem}
\begin{proof}
When \( v = 1 \), the limit matrix \( A \) of \( A(n) \) satisfies the contraction condition. In this case, we can take \( K_n = K \), where \( K \) is the cone contracted by \( A \), as shown in \cite{ibrahim2024positivityproofslinearrecurrences}.\\

We know focus on the case $v>1$. For a fixed \( n \), \( A(n) \) satisfies  condition (\ref{cond:1}). By \cref{thm:Vandergraft}, \( A(n) \) contracts a proper cone \( K_n \). We take \( K_n \) specifically as \( K_n = K(V_{i,j}(n)) \subset \mathbb{R}^d_{>0} \), defined in \cref{eq:Cone}.

For $n\in\mathbb{N}$, \( T_n \) denotes the matrix having for columns the vectors  $V_{i,j}(n)$.
In the following, we prove that \( A(n) K_n \subset K_{n+1} \).

Since \( K_n \) is a cone, by linearity, it suffices to show that the inclusion holds for the vectors in \( S_1(K_n) \), where
$$
S_1(K_n) = \{ W = T_n \alpha \mid \alpha \in C,~\alpha_{1,1} = 1 \}.
$$

Let \( W_n = T_n \alpha \) be a vector in \( S_1(K_n) \). Then \( W_{n+1} = T_{n+1} \alpha \) is a vector in \( S_1(K_{n+1}) \), and its image \( A(n+1) W_{n+1} \) lies in the interior of \( K_{n+1} \).

We introduce the norm  $$\|x\|_n = \|T_n^{-1} x\|,$$ and $\Delta_n$ the distance induced by $\|.\|_n$. As shown by \cref{pro:distance},
$$
\Delta_{n+1}(A(n+1) W_{n+1}, \partial K_{n+1}) \geq \frac{ \lambda_{1,n+1} - |\lambda_{2,n+1}|}{2}.
$$
To prove the inclusion property, it is sufficient to show that
$$
\| A(n) W_n - A(n+1) W_{n+1} \|_{n+1} < \frac{\lambda_{1,n+1} - |\lambda_{2,n+1}|}{2}.
$$
This ensures that \( A(n) W_n \) is contained in \( K_{n+1} \).

By definition, we have
\begin{align*}
    \| A(n) W_n - A(n+1) W_{n+1} \|_{n+1} &= \| (A(n) T_n - A(n+1) T_{n+1}) \alpha \|_{n+1} \\
    &= \| T_{n+1}^{-1} (A(n) T_n - A(n+1) T_{n+1}) \alpha \| \\
    &\leq \| T_{n+1}^{-1} (A(n) T_n - A(n+1) T_{n+1}) \|
\end{align*}
since $\|\alpha\|=1$. Next, we use the following relation for \( T_{n+1}^{-1} \),
\begin{align*}
   T_{n+1}^{-1}-T_n^{-1}= T_{n+1}^{-1}(T_n-T_{n+1})T_{n}^{-1}.
\end{align*}
This leads to the inequality,
\begin{align*} \| A(n) W_n - A(n+1)& W_{n+1} \|_{n+1} \leq \| T_n^{-1} A(n) T_n - T_{n+1}^{-1} A(n+1) T_{n+1} \| \\
    &\quad + \| T_{n+1}^{-1}(T_n-T_{n+1})T_{n}^{-1} A(n) T_n \|.
\end{align*}
Since
$$ 
T_n^{-1} A(n) T_n = \begin{bmatrix}
    \lambda_{1,n} & 0 & \cdots & 0 \\
    0 & \lambda_{2,n} & \cdots & 0 \\
    \vdots & \vdots & \ddots & \vdots \\
    0 & 0 & \cdots & J_{k_n,\varepsilon}
\end{bmatrix},
$$
we obtain first the following bound as $n\to \infty$,
\begin{equation}\label{eq:soustraction}
    \| T_n^{-1} A(n) T_n - T_{n+1}^{-1} A(n+1) T_{n+1} \| = O\left( \max_{i=1,\dots,k} |\lambda_{i,n} - \lambda_{i,n+1}| \right).
\end{equation}
Next, let \( (t_{i,j}) \) denote the entries of \( T_{n+1} - T_n \). Then, by \cref{eq:basis},
$$
t_{i,j} = c_{i,j} (\lambda_{i,n}^{j-1} - \lambda_{i,n+1}^{j-1}),\quad n\to \infty
$$
where \( c_{i,j} \) are constants independent of \( n \). Therefore, the norm of the difference \( T_{n+1} - T_n \) can be bounded as
$$
\| T_{n+1} - T_n \| = O\left( \max_{i=1,\dots,k} |\lambda_{i,n} - \lambda_{i,n+1}| \right),~n \to \infty.$$
Due to the dominance of $\lambda_{1,n}$ and the condition on $\varepsilon$,  
\begin{equation}\label{eq:lambda1}
    \|T_{n}^{-1}A(n)T_{n}\|=\lambda_{1,n}\to \lambda_1,\quad n\to \infty.
\end{equation}
Since \( T_n \to T \) (invertible by condition (\ref{cond:2})), \( \| T_{n+1}^{-1} \| \) is bounded as $n$ tends to infinity.
Combining this with the bound in \cref{eq:lambda1},
\begin{equation}\label{eq:sum}
    \|T_{n+1}^{-1}(T_n-T_{n+1})T_n^{-1}A(n)T_n\|= O\left(\max_{i=1,\dots,k} |\lambda_{i,n} - \lambda_{i,n+1}| \right).
\end{equation}
Finally, by \cref{eq:sum} and \cref{eq:soustraction}, as $n\to\infty$
$$
\| A(n) W_n - A(n+1) W_{n+1} \|_{n+1} = O\left( \max_{i=1,\dots,k} |\lambda_{i,n} - \lambda_{i,n+1}| \right),
$$
which is smaller than $\frac{\lambda_{1,n+1}-|\lambda_{2,n+1}|}{2}$ asymptotically. Indeed,  condition (\ref{cond:3}), asserts $$\underset{i=1,\dots,k}{\max} |\lambda_{i,n}-\lambda_{i,n+1}|=o(\lambda_{1,n}-|\lambda_{2,n}|),\quad n\to \infty.$$ 
Under this assumption, we prove that 
\begin{equation}
    \label{eq:equivalence}\lambda_{1,n+1}-|\lambda_{2,n+1}|\sim \lambda_{1,n}-|\lambda_{2,n}|,\quad n\to\infty.
\end{equation}
Consequently, as $n\to\infty$ 
$$\underset{i=1,\dots,k}{\max} |\lambda_{i,n}-\lambda_{i,n+1}|=o(\lambda_{1,n+1}-|\lambda_{2,n+1}|),$$ thus  \( A_n K_n \subseteq K_{n+1} \) for \( n \) large enough. 

In fact, 
\begin{align*}
     |\lambda_{1,n+1}-\lambda_{1,n}|&=O\Bigl(\underset{i=1,\dots,k}{\max} |\lambda_{i,n}-\lambda_{i,n+1}|\Bigr)\\
     ||\lambda_{2,n}|-|\lambda_{2,n+1}||&\leq |\lambda_{2,n}-\lambda_{2,n+1}|=O\Bigl(\underset{i=1,\dots,k}{\max} |\lambda_{i,n}-\lambda_{i,n+1}|\Bigr).
\end{align*}
Then,
$$\lambda_{1,n+1}-|\lambda_{2,n+1}|=\lambda_{1,n}-|\lambda_{2,n}|+O(\underset{i=1,\dots,k}{\max} |\lambda_{i,n}-\lambda_{i,n+1}|).$$
Thus, by condition (\ref{cond:3}) we obtain the equivalence in \cref{eq:equivalence}.
\end{proof}

\begin{proposition}
    
\label{cor:decidability}
Let \( (u_n)_n \) be a P-finite sequence satisfying the recurrence in \cref{rec}, of Poincaré type. Assume that:  
\begin{itemize}
    \item The recurrence has several dominant eigenvalues, all simple.  
    \item One of these eigenvalues, denoted \( \lambda_1 \), is real and positive.  
    \item The associated operator \( A(n) \) satisfies the conditions of Theorem \ref{theo:cond}.  
    \item \( V_{1} \) is a positive eigenvector of the limit matrix \( A \) associated with \( \lambda_1 \)
    and 
    \[
    \lim_{n \to \infty} \frac{U_n}{\|U_n\|} = V_{1}.
    \]  
\end{itemize}  
Then, the positivity of \( (u_n)_n \) is decidable.  
\end{proposition}

\begin{proof}
We prove the termination of \cref{algorithm:General} under these conditions.

Take \( K_n = K(V_{i,j}(n)) \subset \mathbb{R}^d_{>0} \). By \cref{theo:cond}, there exists \( N \in \mathbb{N} \) such that  
   \[
   A(n) K_n \subset K_{n+1}, \quad \forall n \geq N.
   \]
   Thus, termination depends on verifying the initial conditions, i.e., proving that \( U_n \in K_n \) for sufficiently large \( n \).
   
 First, the vector \( V_{1} \) and  \( \lim\limits_{n\to\infty} V_{1,1}(n) \) are both positive eigenvectors of \( A \) associated with \( \lambda_1 \). Thus, there exists a constant \( \beta > 0 \) such that  
  $$\lim\limits_{n\to\infty} V_{1,1}(n)=\beta~ V_{1}.$$
Since $V_{1,1}(n)$ is in the interior of $K_n$,  we compute its distance to the boundary \( \partial K_n \) that remains strictly positive:
   $$\Delta_n(V_{1,n}, \partial K_n) = \Delta(T_n^{-1} V_{1,n}, T_n^{-1} \partial K_n)=\Delta( (1,0,\dots,0), C_2),$$
    where $C_2 = T_n^{-1}\partial K_n$ is defined in \cref{proof:2}.
  Using the same approach in the proof of \cref{pro:distance}, we get $$\Delta((1,0,\dots,0),C_2)=\frac{1}{2}.$$
  
On the other hand, let $W_n=\beta~U_n/\|U_n\|$. Then,
   \[
   \|V_{1,1}(n) - W_n\|_n \leq \|T_n^{-1}\| (\|V_{1,1}(n) - \beta V_1\| + \|\beta V_1 - W_n\|).
   \]
   Since \( \|T_n^{-1}\| \) is bounded, \( \|V_{1,1}(n) - W_n\|_n \) becomes arbitrarily small as $n\to \infty$. Thus, for $n$ sufficiently large, $U_n$ remains in the cone $K_n$.
   
Returning to \cref{algorithm:General}, once we find $n_0\in\mathbb{N}$ such that $U_{n_0}\in K_{n_0}$, by induction \( U_n \) remains in \( K_n\subset \mathbb{R}^d_{>0}\) for all \( n \geq n_0 \).  Therefore, deciding positivity reduces to checking the sign of the initial values \( U_0, \dots, U_{n_0-1} \).
\end{proof}

\section{Examples}\label{sec:Examples}

We now present examples for which the positivity could not be proven using our previous algorithm.

\begin{itemize}
    \item For \( x \in (0, 1) \), consider the inequality:
\end{itemize}
\[
f(x) = {}_2F_1\left(\frac{1}{2}, \frac{1}{2}; 1; x^2\right) - (1 + x) \, {}_2F_1\left(\frac{1}{2}, \frac{1}{2}; 1; x\right) \geq 0.
\]Let \( (a_n)_{n} \) be the sequence of coefficients in the Taylor expansion of \( f(x) \), i.e., \( f(x) = \sum_{n=0}^{\infty} a_n x^n \). By closure properties, \( (a_n)_{n} \) is a P-finite sequence and satisfies a recurrence of order 6. 
If both sequences \( (u_n)_{n} \) and \( (w_n)_{n} \), defined by \( w_n = a_{2n+1} \) and \( u_n = a_{2n} \), are positive, then the inequality is proven. The sequence \( (w_n)_{n} \) satisfies the recurrence:
\small
\begin{align*}
    64(32n^2 + 24n + 5) & (2n + 3)^2(n + 1)^2 w_{n+1} = \\
    & (32n^2 + 88n + 61)(4n + 3)^2(4n + 1)^2 w_n.
\end{align*}
\normalsize
Thus, positivity is trivial.

On the other hand, \( (u_n)_{n} \) is the solution of a second-order recurrence relation with a double eigenvalue \( \lambda_1 = 1 \). However, the matrices $A(n)$ satisfy the contraction condition. Applying \cref{algorithm:General}, the vectors \( U(n) = (u_n, u_{n+1}) \in K_n \subset \mathbb{R}^2_{>0} \) for all \( n \geq 2 \). Thus, \( (u_n)_{n} \) is positive. 
\begin{itemize}
    \item Let $P_k$ denote the $k$-th Legendre polynomial. The following example is Turàn's inequality for $x=\frac{1}{2}$,
\end{itemize}
$$(P_n^2-P_{n-1}P_{n+1})(\frac{1}{2})\geq 0,\quad n\in\mathbb{N}.$$
The associated sequence $(u_n)_n$ satisfies a recurrence of order $3$ and its characteristic polynomial is $X^3-1$. However, as $n\to\infty$ 
$$\lambda_{1,n}\sim 1-\frac{1}{n},~\lambda_{2,n}\sim \frac{-1+i\sqrt{3}}{2}+ \frac{3(1-i\sqrt3)}{2n},~\lambda_{3,n}=\overline{\lambda_{2,n}}.$$
Thus, the conditions of \cref{cor:decidability} are satisfied and by \cref{algorithm:General} the vectors $U_n\in K_n\subset\mathbb{R}^3_{>0}$ for all $n\geq 8$.

\begin{remark}
    The above inequality holds for all  $-1 \leq x \leq 1$ . Using \cref{algorithm:General} to prove it, we observe that the inclusion index  $N \to \infty$  as  $x$ approaches the boundary points  $-1$  and  $1$.
\end{remark}

\section{Approximations by Asymptotic Truncation}\label{sec:practice}
In practice, computing the cone \( K_n \) in \cref{algorithm:General} is simplified by approximating the eigenvalues \( \lambda_{i,n} \) through asymptotic expansions. Specifically, we truncate the expansions of $\lambda_{i,n}$ as $n\to\infty$. The order of truncation is iteratively increased if necessary. 

For instance, in \cref{sec:Example}, the cone \( K_n \) was constructed by taking an asymptotic truncation of \( \lambda_{i,n} \) up to order \( 1 \):
$$
    \lambda_{1,n} \sim 1+ \frac{-2+\sqrt{2}}{n},~ 
    \lambda_{2,n} \sim 1- \frac{2+\sqrt{2}}{n}, \quad n\to\infty.$$

With these approximations, we construct the associated vectors \( V_{i,j}(n) \), and then the cone \( K_n \) is given by
\[
K_n = K({V_{i,j}(n)}), \quad n > 0.
\]

\begin{remark}
The computations for the steps following the construction of the sequence of cones can be carried out as outlined in \cite[section 6 ]{ibrahim2024positivityproofslinearrecurrences}.
\end{remark}
\section{Conclusion}
We present a method for proving positivity of P-finite sequences of Poincaré type under the assumption that the associated recurrence operator satisfies the contraction condition. This extends the approach in
\cite{ibrahim2024positivityproofslinearrecurrences} to a larger class of recurrences, where the dominant eigenvalue is not necessarily unique or simple.

For recurrences with multiple simple dominant eigenvalues of arbitrary order, we give an asymptotic conditions on the eigenvalues of the recurrence operator that ensure the existence of a sequence of cones,  each mapped into the next by the recurrence operator. Under an assumption on the asymptotic behavior of the sequence, this allows us to decide positivity within this class of recurrences. This assumption depends on the choice of initial conditions and is analogous to the genericity assumption in  \cite{ibrahim2024positivityproofslinearrecurrences,KauersPillwein2010a}.  
Moreover, under this assumption the termination of the method can be detected in advance by checking whether the asymptotic conditions on the eigenvalues are satisfied. This is the main distinction with the Gerhold-kauers method \cite{GerholdKauers2005}.

In other works, analytic techniques were used to prove positivity \cite{MelczerMezzarobba2022,BostanYurkevich2022a}.  Our approach, instead, relies on algebraic tools, as we aim to generalize the cone-based method to recurrences with real parameters.




\bibliographystyle{abbrv} \bibliography{biblio}
\end{document}

%% file: preamble.tex
\newtheorem{theorem}{Theorem}
\newtheorem{proposition}{Proposition}

\theoremstyle{definition}
\newtheorem{definition}{Definition}
\theoremstyle{remark}

\newtheorem{remark}{Remark}

\newcommand{\trsp}[1]{#1^\mathsf{T}} 